\documentclass{article}
\usepackage[margin=2.5cm]{geometry}

\usepackage{authblk}
\usepackage{balance} 
\usepackage{algorithm}
\usepackage{algorithmicx}
\usepackage[noEnd=true]{algpseudocodex}
\usepackage{amsmath,amsfonts,graphicx,bm,amsthm,amssymb}
\usepackage{natbib}
\allowdisplaybreaks

\usepackage{multirow}
\usepackage{makecell}
\usepackage{tabularx}
\usepackage{comment}
\usepackage{graphicx}

\newcommand{\keywords}[1]{%
  \par\noindent\textbf{Keywords: } #1
}

\newcommand{\APS}{\mathsf{APS}}
\newcommand{\MMS}{\mathsf{MMS}}
\newcommand{\WMMS}{\mathsf{WMMS}}

\newtheorem{lemma}{Lemma}
\newtheorem{corollary}{Corollary}
\newtheorem{theorem}{Theorem}
\newtheorem{definition}{Definition}
\newtheorem{proposition}{Proposition}

\title{On the Fair Allocation to Asymmetric Agents with Binary XOS Valuations}

\author{
Ziheng Chen$^{1,2,}$\thanks{These authors contributed equally to this work.} ,
Bo Li$^{3,*}$,
Zihan Luo$^{1,2,*}$,
Jialin Zhang$^{1,2,*,}$\thanks{Corresponding author. zhangjialin@ict.ac.cn} \\
{\small
$^{1}$ Institute of Computing Technology, Chinese Academy of Sciences \\
$^{2}$ University of Chinese Academy of Sciences \\
$^{3}$ Department of Computing, The Hong Kong Polytechnic University
}
}

\newcommand{\BibTeX}{\rm B\kern-.05em{\sc i\kern-.025em b}\kern-.08em\TeX}

\begin{document}


\maketitle 


\begin{abstract}
We study the problem of allocating $m$ indivisible goods among $n$ agents, where each agent's valuation is fractionally subadditive (XOS). 
With respect to AnyPrice Share (APS) fairness, \citet{KulkarniKM24} showed that, when agents have binary marginal values, a $0.1222$-APS allocation can be found in polynomial time, and there exists an instance where no allocation is better than $0.5$-approximate APS.
Very recently, \citet{feige2025fair} extended the problem to the asymmetric case, where agents may have different entitlements, and improved the approximation ratio to $1/6$ for general XOS valuations.
In this work, we focus on the asymmetric setting with binary XOS valuations, and further improve the approximation ratio to  $1/2$, which matches the known upper bound. We also present a polynomial-time algorithm to compute such an allocation.
Beyond APS fairness, we also study the weighted maximin share (WMMS) fairness. \citet{farhadi2019fair} showed that, a $1/n$-WMMS allocation always exists for agents with general additive valuations, and that this approximation ratio is tight. We extend this result to general XOS valuations, where a $1/n$-WMMS allocation still exists, and this approximation ratio cannot be improved even when marginal values are binary. This shows a sharp contrast to binary additive valuations, where an exact WMMS allocation exists and can be found in polynomial time.
\end{abstract}




\section{Introduction}

In this paper, we study the fair allocation of $m$ indivisible goods among $n$ agents.
A central and widely accepted fairness criterion in this problem is the maximin share (MMS), introduced by \citet{budish2011combinatorial}. However, prior work has shown that there exist instances where no allocation can guarantee MMS fairness \citep{feige2021tight,procaccia2014fair,kurokawa2018fair}. 
This raises a fundamental question in fair allocation: to what extent can MMS fairness be guaranteed to be satisfiable? 
This question has inspired a substantial body of work, and for additive, submodular, and XOS valuation functions, constant-factor approximations to MMS fairness have been established; see, for example, \citep{akrami2024breaking,uziahu2023fair,akrami2023randomized}.
\citet{babaioff2024fair} proposed an alternative fairness notion, AnyPrice share (APS), and proved the existence of a constant-approximate APS allocation under additive valuations, which was later generalized to submodular valuations \cite{uziahu2023fair}.

In addition to research on general classes of valuations, significant attention has been given to valuations with binary marginals, such as matroid-rank valuations (i.e., submodular valuations with binary marginals). 
On one hand, these binary functions are well-suited for modeling real-world scenarios -- such as the fair allocation of public housing units \cite{barman2021approximating} -- and make it easier for agents to express their preferences. 
On the other hand, they allow for stronger fairness guarantees. 
For instance, under matroid-rank valuations, an exact MMS allocation is always guaranteed to exist and can be efficiently computed \cite{barman2021approximating,viswanathan2023general}. 
Furthermore, \citet{KulkarniKM24} showed that in this setting, APS is exactly equal to MMS, ensuring that exact APS allocations are always attainable.
This naturally leads to the question: what about binary XOS functions, which generalize binary submodular functions within the complement-free hierarchy? \citet{KulkarniKM24} showed that, in this case, MMS and APS may not be the same, but they satisfy $\text{MMS} \le \text{APS} \le 2\text{MMS}+1$.
Given that the algorithm by \citet{li2021fair} guarantees a $0.366$-approximate MMS allocation, this relationship immediately implies a $0.122$-approximate APS allocation.
Unfortunately, \citet{KulkarniKM24} also showed that it is impossible to achieve an approximation better than $0.5$ for APS under binary XOS valuations.

Very recently, \citet{feige2025fair} improved upon previous results by showing that a $\frac{4}{17}(\approx 0.235)$-approximate APS allocation exists under the general XOS valuations.
They also extended their analysis to a more general setting where agents are asymmetric—that is, agents have different entitlements to the goods, with those holding higher entitlements expected to receive a larger share. This asymmetric setting is particularly relevant to real-world applications. 
As discussed in \cite{farhadi2019fair}, in many religions, cultures, and legal systems, inherited wealth is often distributed unequally, and the allocation of natural resources between neighboring countries is frequently based on geographic, economic, or political considerations.
\citet{feige2025fair} proved that, for general XOS valuations in the asymmetric setting, a $\frac{1}{6}$-approximate APS allocation is guaranteed to exist.

\begin{table*}[t]
    \centering
    \renewcommand{\arraystretch}{1.6}
    \resizebox{\textwidth}{!}{
    \begin{tabular}{|c|cc|cc|}
        \hline
        \multirow{2}*{Setting} & \multicolumn{2}{c|}{APS} & \multicolumn{2}{c|}{WMMS}  \\
        \cline{2-5}
        & Symmetric & Asymmetric& Symmetric & Asymmetric\\
        \hline
        Additive & \makecell{$\frac{3}{5}$ \\ (\citet{babaioff2024fair})} & \makecell{$\frac{3}{5}$ \\ (\citet{babaioff2024fair})} & \makecell{$\frac{10}{13}$ \\ (\citet{heidari2025improved})} & \makecell{$\frac{1}{n}$ \\  (\citet{farhadi2019fair}) }  \\
        \hline
        \makecell{{\bf Binary} \\ {\bf XOS}} & \makecell{ $0.1222 \xrightarrow[]{\text{\bf Theorem \ref{APStheorem6}}} \frac{1}{2}$ \\ (\citet{KulkarniKM24})} & \makecell{$\frac{1}{2}$\\ ({\bf Theorem \ref{APStheorem6}})} & \makecell{$0.3666\xrightarrow[\text{\citet{hummel2025maximin}}]{\text{\bf Corollary \ref{MMScorollary}}} \frac{1}{2}$ \\ (\citet{KulkarniKM24})} & \makecell{ $\frac{1}{n}$ \\ ({\bf Theorem \ref{WMMStheorem1}})} \\
        \hline
        \makecell{General\\ XOS} & \makecell{ $\frac{4}{17}$ \\ (\citet{feige2025fair})} & \makecell{$\frac{1}{6}$ \\ (\citet{feige2025fair})} & \makecell{ $\frac{4}{17}$ \\ (\citet{feige2025fair})} & \makecell{$\frac{1}{n}$ \\ ({\bf Theorem \ref{WMMStheorem1}}) }\\
        \hline
    \end{tabular}}
    \caption{Main results. We remark that all our approximations are tight. For clarity, we provide a comparison with the best-known approximations in related settings.}
    \label{tab:main-results}
\end{table*}

While the results of \citet{feige2025fair} have significantly advanced our understanding of approximate MMS and APS fairness under XOS valuations, several open questions from \cite{KulkarniKM24} remain unresolved.
First, for binary XOS valuations, the optimal approximation ratio is still unknown, which is known to be between $\frac{4}{17}$ and $\frac{1}{2}$ for the symmetric setting, and between $\frac{1}{6}$ and $\frac{1}{2}$ for the asymmetric setting.
Second, in the asymmetric setting, the relationship between APS and weighted MMS (WMMS, as defined in \cite{farhadi2019fair}), is still unknown, and the optimal approximation ratio of WMMS is unknown. 
In this paper, we answer the above two questions.



\subsection{Our Contribution}

In this paper, we build upon the work of \cite{KulkarniKM24} to investigate the fair allocation of $m$ indivisible goods among $n$ agents, where agents' valuations are binary XOS.
We focus on APS and WMMS fairness, and the general setting when agents have different entitlements. 
The main results and their comparisons to the related work are summarized in Table \ref{tab:main-results}.

Our contributions are twofold. First, for APS fairness, we show that a $\frac{1}{2}$-approximate APS allocation always exists, and thus this approximation ratio is tight. 
Our result applies to the general asymmetric case; see Theorem \ref{APStheorem5}. 
In the symmetric setting, since APS is always at least as large as MMS, our result also guarantees the existence of a $\frac{1}{2}$-approximate MMS allocation. 
Furthermore, we provide a polynomial-time implementation to compute such an allocation; see Theorem \ref{APStheorem6}.

In the asymmetric setting, \citet{farhadi2019fair} first introduced WMMS fairness to account for agents with different weights. 
In this context, APS is no longer an upper bound for WMMS; indeed, there are instances (with agents having general additive valuations) where $\frac{\APS}{\WMMS}$ is infinitesimal for some agent; see Proposition \ref{WMMSAPS}. \citet{farhadi2019fair} further demonstrated that a round-robin algorithm guarantees a $\frac{1}{n}$-approximate WMMS allocation for additive valuations.
Extending beyond additive valuations, we generalize the result of \citet{farhadi2019fair} by proving that a $\frac{1}{n}$-approximate WMMS allocation always exists for arbitrary XOS valuations (Theorem \ref{WMMStheorem1}). 
Nevertheless, we show that for XOS valuations -- even when restricted to binary cases -- a better than $\frac{1}{n}$-approximate WMMS allocation cannot be guaranteed (Theorem \ref{WMMSexample}).
In contrast, we prove that an exact WMMS allocation is achievable for binary additive valuations, underscoring the inherent challenges of attaining WMMS fairness for binary XOS valuations (Theorem \ref{WMMSAddAlgo}).


\subsection{Related Work}



MMS fairness was first introduced in \citet{budish2011combinatorial} as a relaxation of proportionality \cite{steinhaus1949division}. However, there are instances for which no MMS allocation exists \cite{feige2021tight,procaccia2014fair,kurokawa2018fair}. 
For additive valuations, \citet{kurokawa2018fair} proved the existence of $\frac{2}{3}$-MMS allocations. Subsequently, \citet{ghodsi2018fair} improved the approximation ratio to $\frac{3}{4}$, which was further improved to $\frac{3}{4} + o(1)$ by \citet{garg2020improved,akrami2023simplification}, $\frac{3}{4} + \frac{3}{3836}$ by \citet{akrami2024breaking}, and most recently to $\frac{10}{13}$ by \citet{heidari2025improved}.
For submodular valuations, 
\citet{barman2020approximation} and \citet{ghodsi2018fair}, respectively, designed algorithms to compute $0.21$-MMS and $\frac{1}{3}$-MMS allocations.  \citet{ghodsi2018fair} showed that a better than $\frac{3}{4}$-MMS allocation may not exist.
Recently, \citet{uziahu2023fair} improved the approximation ratio to $\frac{10}{27}$.
For XOS valuations, \citet{ghodsi2018fair} proved that the best possible approximation ratio is between $\frac{1}{5}$ and $\frac{1}{2}$. 
Later, \citet{akrami2023randomized} designed an algorithm computing $0.23$-MMS allocations.
For general subadditive valuations, approximate MMS fair algorithms are given in \cite{ghodsi2018fair, seddighin2024improved,feige2025concentration,seddighin2025beating, feige2025multi}.
Finally, when agents have asymmetric entitlements, \citet{farhadi2019fair} introduced weighted MMS fairness, and proved that the best possible approximation ratio for additive valuations is $\frac{1}{n}$.

\citet{babaioff2024fair} introduced APS fairness, which addresses some of the modeling concerns of MMS, especially when agents have different weights.
They designed a $0.667$-APS algorithm for additive valuations. 
Later, \citet{uziahu2023fair} extended the setting to submodular valuations and presented a $\frac{1}{3}$-approximation algorithm.
Recently, \citet{feige2025fair} showed the existence of $\frac{1}{6}$-APS allocations for XOS valuations, and improved the approximation ratio to $\frac{4}{17}$-APS allocation when agents have equal entitlements.

Although all the above work considers the allocation of goods, the minor problem of chores has also been widely studied. 
For example, with equal entitlements, the MMS fairness is studied in \cite{aziz2017algorithms,DBLP:conf/sigecom/HuangL21,DBLP:conf/sigecom/HuangS23} for additive valuations and in \cite{DBLP:conf/nips/0037WZ23} for more general valuations.
With asymmetric entitlements, the WMMS fairness is studied in \cite{aziz2019maxmin,wang2024improved}, and APS fairness is studied in \cite{babaioff2024fair,li2022almost,feige2023picking}.

\section{Preliminaries}

We study the problem of fairly allocating a set of $m$ indivisible goods among $n$ agents. 
Let $[k] = \{1, 2, \dots, k\}$ for a positive integer $k$.
Denote the set of goods by $M:=\{g_1, g_2, \dots, g_m\}$ and the set of agents by $N:=\{a_1, a_2, \dots, a_n\}$. The preference of each agent $a_i \in N$ is defined by a valuation function $v_i: 2^{M} \rightarrow \mathbb{R}_{\geq 0}$ over the set of goods. 
Specifically, $v_i(S)$ is the value that $a_i$ has for the subset of goods $S \subseteq M$. Furthermore, we emphasize the imbalance between agents with their entitlements $(b_i)_{i\in [n]}$. Therefore, we represent a fair allocation instance by the quadruple $(N,M,(v_i)_{i \in n}, (b_i)_{i\in [n]})$.

An allocation, $A := (A_1, \ldots, A_n)$, is a partition of all the goods among the $n$ agents, i.e., $A_i \cap A_j = \varnothing$ for all $i \neq j$ and $\bigcup_{i \in [n]} A_i = M$. We also define a partial allocation, denoted by $P = (P_1, \ldots, P_n)$, as a partition of any subset of goods, where $P_i \cap P_j = \varnothing$ for all $i \neq j$ and $\bigcup_{i \in [n]} P_i \subseteq M$.

\subsection{Fairness Notions}

For the sake of simplicity, in a fair allocation instance $(N,M,(v_i)_{i \in n}, (b_i)_{i\in [n]})$, let
\begin{equation*}
    \Pi_n = \{S = (S_i)_{i\in [n]} \mid M = \cup_{i\in [n]} S_i, S_i\cap S_j = \varnothing, i\ne j\}
\end{equation*}
be the set of all partitions of $M$ to $n$ bundles.

The maximin share is an analogy of the cake-cutting problem, where the cutter could only expect to obtain the minimum share of the cake. Therefore she need to find the allocation that would maximize the minimum share. In the unweighted version of maximin share, we treat all agents equally. Namely, each agent has the same entitlement $\frac{1}{n}$.

\begin{definition}[Maximin Share]\label{defMMS}
    For a fair allocation instance $(N,M,(v_i)_{i \in [n]}, (\frac{1}{n})_{i\in [n]})$, the maximin share (MMS) of $a_i$ is
    \begin{equation*}
        \MMS_i^N (M) = \max_{S\in \Pi_n} \min_{j\in [n]} v_i(S_j).
    \end{equation*}
\end{definition}

As an extension, the weighted maximin share takes the entitlements into consideration. Now the cutter indeed proposes an allocation according to her valuation function and the agents' entitlements instead of just a partition.

\begin{definition}[Weighted Maximin Share]\label{defWMMS}
    For a fair allocation instance $(N,M,(v_i)_{i \in [n]}, (b_i)_{i\in [n]})$, the weighted maximin share (WMMS) of $a_i$ is
    \begin{equation*}
        \WMMS_i^N (M) = \max_{S\in \Pi_n} \min_{j\in [n]} v_i(S_j) \frac{b_i}{b_j}.
    \end{equation*}
\end{definition}

It is obvious that when all entitlements are equal, the weighted maximin share degenerates to the maximin share. Furthermore, we call the partition that reaches the WMMS value the WMMS partition.

\begin{definition}[WMMS Partition] \label{defWMMSpart}
    For a fair allocation instance $(N,M,(v_i)_{i \in [n]}, (b_i)_{i\in [n]})$, a WMMS partition of $a_i$ is
    \begin{equation*}
        S^i \in \mathop{\arg\max}_{S\in \Pi_n} \min_{j\in [n]} v_i(S^i_j) \frac{b_i}{b_j},
    \end{equation*}
    where $S^i_j$ is the bundle $a_i$ wants to assign to $a_j$.
\end{definition}

Informally, AnyPrice Share indicates the value that an agent $a_i$ can guarantee herself at any given price of items, based on her entitlement $b_i$.

\begin{definition}[AnyPrice Share]
Let $\mathcal{P}$ represent the price vector simplex corresponding to the set of goods $M$, or formally 
\begin{equation*}
    \mathcal{P} = \left\{p = (p_1, p_2, \ldots, p_m) \geq 0 \bigm| \sum_{i\in [m]} p_i = 1 \right\}.
\end{equation*}
For a fair allocation instance $ (N, M, (v_i)_{i \in [n]},(b_i)_{i \in [n]})$, the APS value of $a_i$ is
\begin{equation*}
    \APS_i^N (M):=\min_{p\in \mathcal{P}}\max_{S\subseteq M,p(S)\leq b_i}v_i(S),
\end{equation*}
where $p(S)$ is the sum of prices of goods in $S$.
\end{definition}

For brevity we will refer to $(\mathsf{W})\MMS_i^N(M)$ as $(\mathsf{W})\MMS_i$ and $\APS_i^N(M)$ as $\APS_i$ if the instance is clear by the context.


We emphasize that in the definition of (weighted) maximin share and AnyPrice share of some $a_i$, only her own value function counts. The other agents' participation is limited to their entitlements.

An allocation $A$ is $\alpha$-APS (similarly, $\alpha$-(W)MMS) if and only if $v_i(A_i) \geq \alpha \APS_i$ (similarly, $v_i(A_i) \geq \alpha (\mathsf{W})\MMS_i$) for every agent $i \in N$.

\subsection{Valuation Functions}

This paper focuses on XOS valuations with binary marginals (or binary XOS valuations). 

\begin{definition}[Additive Function]
A function $v: 2^{M} \to \mathbb{R}_{\geq 0}$ is additive if the value of a set of goods is equal to the sum of the values of the goods in the set, that is, $v(S) = \sum_{g \in S} v(\{g\})$.
\end{definition}

\begin{definition}[XOS Function]
A function $v: 2^{M} \to \mathbb{R}_{\geq 0}$ is XOS, or fractionally subadditive,  iff there exists a collection of additive functions $\{l_t\}_{t\in [L]}$ such that, for every subset $S\subseteq M$, we have $v(S)=\max_{t\in [L]}l_t(S)$. Notice that the number of additive functions, $|L|$, can be exponentially large in $m$.
\end{definition}

Given an XOS function $v$, we suppose there is a query oracle $\mathcal{O}$ that receives a set $S$ as input and computes $v(S)$ in time $O(1)$.

Furthermore, a valuation $v$ has binary marginals if and only if $v(S \cup \{g\}) - v(S) \in \{0, 1\}$ for any subset of goods $S \subseteq M$ and any goods $g \in M$. As a direct consequence, the valuations we consider are monotonic: $v(S) \leq v(T)$ for any subsets $S \subseteq T \subseteq M$. In addition, we require the valuations to satisfy $v_i(\varnothing) = 0$ for each $i \in [n]$.

When valuation $v$ has binary marginals, a set of goods $S\subseteq M$ is said to be non-wasteful respect to valuation $v$, iff $v(S) = |S|$. \citet{barman2021approximating} proved that for any binary XOS function $v$ and every set $S\subseteq M$, there exists a non-wasteful subset $X\subseteq S$ with the property that $v(X) = |X| = v(S)$, and the non-wasteful set $X$ can be computed efficiently in time $O(m)$.

\section{APS with Binary XOS Valuations}

In this section, we focus on computing approximately APS allocations under binary XOS valuations. 
We first prove the existence of a $\frac{1}{2}$-APS allocation in Section \ref{sec:aps:existence}, and then give a polynomial time implementation in Section \ref{sec:aps:polytime}.
Combining with the upper bound in \cite{KulkarniKM24}, our result is tight. 
Since APS is upper-bounded by MMS when agents have equal entitlements, our result implies a polynomial-time algorithm computing $\frac{1}{2}$-MMS allocations in the symmetric setting. 



\subsection{The Existence of $\frac{1}{2}$-APS Allocations}
\label{sec:aps:existence}
Our proof is constructive, and the algorithm is formally described in Algorithm \ref{APSexist5}. The algorithm returns a $\frac{1}{2}$-APS allocation if the APS values of all the agents are known. To make it a subprocess of our polynomial-time algorithm \ref{APSpoly5}, we let Algorithm \ref{APSexist5} output an unsatisfied agent if the input guessing values of APS' are too large.

Suppose the APS' of the agents are $s_1, s_2,\ldots,s_n$, then we can rename the agents such that 
\begin{align}\label{eq:APS:existence:1}
    \frac{\lceil\frac{1}{2}s_1\rceil}{b_1}\le \frac{\lceil\frac{1}{2}s_2\rceil}{b_2}\le \dots\le \frac{\lceil\frac{1}{2}s_n\rceil}{b_n}.
\end{align}
We prove that we can allocate each agent a non-wasteful bundle with value at least half of their APS.

\begin{lemma} \label{APSequalprice5}
    For some agent $a_i$ in a fair allocation instance $(N,M,(v_i)_{i \in [n]}, (b_i)_{i\in [n]})$ with binary XOS valuations, suppose $S\subseteq M$ is a subset of goods, then there exists a non-wasteful bundle $B\subseteq M\setminus S$, satisfying $v_i(B) = |B| = \APS_i - \lfloor b_i|S|\rfloor$.
\end{lemma}

\begin{proof}
    If $S = \varnothing$, the lemma is straightforward according to the definition of APS. Otherwise, consider the following price for all goods $l$,
    \begin{equation*}
        p_l = \begin{cases}
            \frac{1}{|S|}, &l\in S, \\
            0, &l\notin S.
        \end{cases}
    \end{equation*}
    With her entitlement $b_i$, $a_i$ could at most afford $\lfloor b_i|S|\rfloor$ goods in $S$. While by the definition of APS, $a_i$ could always afford at least $\APS_i$ value of items in $M$ with any price. Therefore, for remaining goods, we have $v_i(M\setminus S) \geq \APS_i - \lfloor b_i|S|\rfloor$. Then by \citet{barman2021approximating}, we could strip some goods off $M\setminus S$ to get the demanded non-wasteful bundle $B$, where $v_i(B)=|B| = \APS_i - \lfloor b_i|S|\rfloor$.
\end{proof}

\begin{algorithm}
\caption{$\frac{1}{2}$-APS Existence} \label{APSexist5}
\begin{algorithmic}[1]
\Require An instance $(N,M,(v_i)_{i\in [n]},(b_i)_{i\in [n]})$ with binary XOS valuations, where the APS values $s_1, s_2,\ldots,s_n$  satisfy Inequality \ref{eq:APS:existence:1}.
\Ensure If there exists a $\frac{1}{2}$-APS allocation, return this allocation $A=(A_1,A_2, \ldots, A_n)$. Otherwise, return an unsatisfied agent $a_i$.
    \For {$i\gets 1, 2, \dots, n$}
        \If {$v_i(M)\ge \lceil\frac{1}{2}s_i\rceil$}
            \State Assign $a_i$ a non-wasteful bundle $A_i$ from $M$, where $v_i(A_i) = |A_i| = \lceil\frac{1}{2}s_i\rceil$. \label{APSexistalg}
            \State Set $M\gets M\setminus A_i$.
        \Else\label{non-wasteful}
            \State \Return $a_i$
        \EndIf
    \EndFor
    \If {$i=n$ and $M\ne \varnothing$}
        \State Allocate the remaining items arbitrarily.
        \State\Return $A = (A_1, A_2, \dots, A_n)$.
    \EndIf
    
\end{algorithmic}
\end{algorithm}

\begin{theorem} \label{APStheorem5}
    When the estimated APS values are accurate for all agents, which means $s_i=\APS_i$ for $i\in [n]$, Algorithm \ref{APSexist5} computes a $\frac{1}{2}$-APS allocation for each agent $a_i$ of the fair allocation instance $(N,M,(v_i)_{i\in [n]},(b_i)_{i\in [n]})$, where $(v_i)_{i\in [n]}$ are binary XOS valuations.
\end{theorem}

\begin{proof}
    In Algorithm \ref{APSexist5} we claimed the existence of $(A_i)_{i\in [n]}$ in line \ref{APSexistalg} when $s_i=\APS_i$ for $i\in [n]$. Now we prove by induction. For $a_1$, it is obvious that there is a bundle $B_1$ where $v_1(B_1) = |B_1| = \APS_1 \ge \lceil\frac{1}{2}\APS_1\rceil$. Then similar to the proof of lemma \ref{APSequalprice5}, by \citet{barman2021approximating},  we can find a non-wasteful bundle $A_1\subseteq B_1$ such that $v_1(A_1) = |A_1|=\lceil \frac{1}{2} \APS_1\rceil$.
     
    We then assume that in the first $k$ iterations where $1\le k \le n - 1$, $(a_i)_{i\in [k]}$ have received their bundles $(A_i)_{i\in [k]}$ satisfying $v_i(A_i)=|A_i| = \lceil\frac{1}{2}\APS_i\rceil$. We will prove that $a_{k+1}$ could obtain a bundle $A_{k+1}$ where $v_{k+1}(A_{k+1})=|A_{k+1}| = \lceil\frac{1}{2}\APS_{k+1}\rceil$.
    
    Let $S_{k + 1} = \cup_{i = 1}^k A_i$ be the set of items that has been taken by the first $k$ agents. Substitute $S$ to $S_{k + 1}$ in lemma \ref{APSequalprice5}, we know that $v_{k + 1}(M\setminus S_{k + 1})\ge \APS_{k+1}-\lfloor b_{k+1}|S_{k+1}|\rfloor$. Therefore, if $\APS_{k+1}-\lfloor b_{k+1}|S_{k+1}|\rfloor\ge \lceil\frac{1}{2}\APS_{k+1}\rceil$, $a_{k+1}$ could receive a bundle $A_{k + 1}$ as demanded.
    \begin{align}
        &\APS_{k+1} - \left\lfloor b_{k+1}|S_{k+1}| \right\rfloor \notag \\
        =& \APS_{k+1}-\left\lfloor b_{k+1}\sum_{i=1}^k|A_{i}|\right\rfloor \notag \\
        =& \APS_{k + 1} - \left\lfloor b_{k+1}\sum_{i=1}^k \left\lceil \frac{1}{2}\APS_i\right\rceil \right\rfloor \notag \\
        =&  \APS_{k + 1} - \left\lfloor b_{k+1}\sum_{i=1}^k \frac{\left\lceil \frac{1}{2}\APS_i\right\rceil} {b_i}b_i\right\rfloor \notag \\
        \ge& \APS_{k + 1} - \left\lfloor b_{k+1}\sum_{i=1}^k \frac{\left\lceil \frac{1}{2}\APS_{k+1}\right\rceil} {b_{k+1}}b_i\right\rfloor \label{order}\\
        =& \APS_{k + 1} - \left\lfloor \left(\sum_{i=1}^k b_i\right)\left\lceil \frac{1}{2}\APS_{k+1}\right\rceil \right\rfloor \notag \\
        \ge& \APS_{k + 1} - \left\lfloor \frac{1}{2}\APS_{k+1} \right\rfloor \label{entitlement} \\
        =& \left\lceil \frac{1}{2}\APS_{k + 1} \right\rceil. \notag 
    \end{align}

    The inequality \ref{order} holds since agents are sorted in ascending order by $\frac{\lceil \frac{1}{2}\APS_{i}\rceil} {b_{i}}$ for $1\le i \le n$. The inequality \ref{entitlement} holds since $k\le n-1$ and therefore $\sum_{i=1}^k b_i<1$.

    By induction, we have proved that each $a_i\in N$ can take a non-wasteful bundle $A_{i}$ where $v_{i}(A_{i}) =|A_{i}|= \lceil \frac{1}{2} \APS_{i} \rceil$.
\end{proof}

\subsection{A Polynomial-time Algorithm}
\label{sec:aps:polytime}
Algorithm \ref{APSexist5} assumes that we are given the APS values of all agents.
However, it is unclear whether APS values can be computed in polynomial time.
In the following, we provide a polynomial-time implementation, as described in Algorithm \ref{APSpoly5}.
Intuitively, Algorithm \ref{APSpoly5} iteratively guesses the APS values, beginning with their trivial upper bounds $s_i = \lfloor b_i m\rfloor$, and invokes Algorithm \ref{APSexist5} to check if a desired allocation is returned. 
If not, an unsatisfied agent $a_i$ is returned, indicating that the current guess for $a_i$'s APS value is too high. Consequently, $s_i$ is decreased, and the process continues in the next iteration.



\begin{algorithm}
\caption{$\frac{1}{2}$-APS Polynomial-Time Algorithm} \label{APSpoly5}
\begin{algorithmic}[1]
\Require An instance with binary XOS valuations $(N,M,(v_i)_{i\in [n]},(b_i)_{i\in [n]})$
\Ensure An $\frac{1}{2}$-APS allocation $A=(A_1, A_2, \ldots, A_n)$
    
    \State Set $M_0 \gets M$.
    \For {$i\gets 1, 2, \dots, n$}
        \State Initialize $s_i = \lfloor b_i m\rfloor$. \Comment{$Upper\ bounds\ of\ APS$}
    \EndFor
    \While {$A=(A_1, A_2, \dots, A_n)$ have not been found}\label{while-loop}
        \State Rename the agents so that 
        \[\frac{\lceil\frac{1}{2}s_1\rceil}{b_1}\le \frac{\lceil\frac{1}{2}s_2\rceil}{b_2}\le \dots\le \frac{\lceil\frac{1}{2}s_n\rceil}{b_n}.
        \]\label{sort} 
        \State Execute Algorithm \ref{APSexist5} with input $s_1,s_2,\ldots,s_n$.
        \If {Algorithm \ref{APSexist5} returns an unsatisfied agent $a_i$}
        \newline\Comment{Some $s_i$ is still too large.}
            \State Set $s_i\gets s_i - 1$.\label{decrease a}
            \State Reset $M\gets M_0$.
        \Else
            \State \Return $A=(A_1, A_2, \dots, A_n)$
        \EndIf
    \EndWhile
\end{algorithmic}
\end{algorithm}


\begin{lemma}\label{APSpoly}
    The time complexity of Algorithm \ref{APSpoly5} is $O(mn(m + \log n))$.
\end{lemma}

\begin{proof}
Each time the while loop in line \ref{while-loop} is executed, either a $\frac{1}{2}$-APS allocation is found, or the guessing value $s_i$ for some agent $a_i$ is decreased by $1$. The sum of the initial guessing APS values for all agents is $\sum_{i=1}^n s_i=\sum_{i=1}^n \lfloor b_i m \rfloor\le m$, so the while loop in line \ref{while-loop} can be executed at most $m$ times.

 During each iteration of the while loop, Algorithm \ref{APSexist5} executes once, since it needs to find at most $n$ times of non-wasteful bundles. If Algorithm $\ref{APSexist5}$ failed to find enough bundles, Algorithm \ref{APSpoly5} will modify the guess and rearrange the sequence of agents in the next iteration. \citet{barman2021approximating} proved that for a binary XOS function $v$ and its ground set $M$, a non-wasteful bundle with value no more than $v(M)$ can be computed in time $O(m)$. While the rearrangement process in line \ref{sort} is done by sorting $n$ numbers, the time complexity of sorting is known as $O(n\log n)$. The sorting ensures the order of the input of Algorithm \ref{APSexist5}.

We have completed our proof that Algorithm \ref{APSpoly5} requires finding non-wasteful bundles for at most $mn$ times and at most $m$ sorting operations. In total, the time complexity of Algorithm \ref{APSpoly5} is $O(mn(m + \log n))$.
\end{proof}

\begin{theorem}\label{APStheorem6}
    Algorithm \ref{APSpoly5} returns a $\frac{1}{2}$-APS allocation $A$ in polynomial-time for any fair allocation instance $(N,M,(v_i)_{i\in [n]},(b_i)_{i\in [n]})$, where $(v_i)_{i\in [n]}$ are binary XOS valuations.
\end{theorem}

\begin{proof}
    In the end of Algorithm \ref{APSpoly5}, each agent $a_i$ gets a bundle with value $\lceil\frac{1}{2}s_i\rceil$. To show that it is a $\frac{1}{2}$-APS allocation, we will prove that throughout this algorithm we have $s_i\ge \APS_i$ for each agent $a_i$.
 
    Suppose that in some iteration, some agent $a_{k+1}$ cannot obtain a bundle $A_{k+1}$ that satisfies $v_{k+1}(A_{k+1}) = \lceil\frac{1}{2} s_{k+1}\rceil$, we claim that $s_{k+1}>\APS_{k+1}$. By Algorithm \ref{APSpoly5} and Lemma \ref{APSequalprice5}, we have
    \begin{equation*}
        \APS_{k + 1} - \left\lfloor b_{k+1} \sum_{i = 1}^k \left\lceil \frac{1}{2} s_{i} \right\rceil \right\rfloor < \left\lceil\frac{1}{2} s_{k + 1}\right\rceil.
    \end{equation*}
    Therefore,
    \begin{align}
        \APS_{k + 1} &< \left\lfloor b_{k+1} \sum_{i = 1}^k \left\lceil \frac{1}{2} s_{i} \right\rceil \right\rfloor + \left\lceil\frac{1}{2} s_{k + 1}\right\rceil \notag \\
        &\le \left\lfloor b_{k+1} \sum_{i = 1}^k \frac{\left\lceil  \frac{1}{2} s_i \right\rceil} {b_i} b_i \right\rfloor + \left\lceil\frac{1}{2} s_{k + 1}\right\rceil \notag \\
        &\le \left\lfloor b_{k+1} \sum_{i = 1}^k \frac{\left\lceil  \frac{1}{2} s_{k+1} \right\rceil} {b_{k+1}} b_i \right\rfloor + \left\lceil\frac{1}{2} s_{k + 1}\right\rceil \label{inequality1}\\
        &\le \left\lfloor \sum_{i = 1}^k b_i \left\lceil \frac{1}{2} s_{k+1} \right\rceil \right\rfloor + \left\lceil \frac{1}{2} s_{k + 1}\right\rceil \notag \\
        &\le \left\lceil \frac{1}{2} s_{k + 1} \right\rceil - 1 + \left\lceil\frac{1}{2} s_{k + 1}\right\rceil \label{inequality2}\\
        &\le s_{k + 1}.\notag 
    \end{align}

    Inequality \ref{inequality1} holds since $\frac{\lceil\frac{1}{2}s_i\rceil}{b_i}\le \frac{\lceil\frac{1}{2}s_{k+1}\rceil}{b_{k+1}}$ when $i<k+1$, and inequality \ref{inequality2} holds since $k<n$ so that $\sum_{i=1}^k b_i<1$.
    
    Therefore, when some agent $a_i$ could not collect a $\frac{1}{2}$-APS bundle from the remaining goods, it holds that $s_{i} > \APS_{i}$. This indicates that the decrease of $s_i$ in line \ref{decrease a} would keep $s_{i}\ge \APS_{i}$ for every $a_i$. So in the end of Algorithm \ref{APSpoly5}, $v_i(A_i) = \lceil\frac{1}{2} s_i\rceil \ge \lceil\frac{1}{2} \APS_i\rceil$, which shows that Algorithm \ref{APSpoly5} is $\frac{1}{2}$-APS.
\end{proof}

Theorem \ref{APStheorem6} implies the existence of $\frac{1}{2}$-MMS. \citet{ghodsi2018fair} provided an instance with binary XOS valuations to show that the upper bound of MMS approximation ratio is $\frac{1}{2}$. We improve the lower bound to $\frac{1}{2}$, which makes the bounds match.

\begin{corollary}\label{MMScorollary}
There exists a $\frac{1}{2}$-MMS allocation under binary XOS valuations, which can be computed in polynomial time.
\end{corollary}

\begin{proof}
    \citet{babaioff2024fair} proved that when all agents have equal entitlements $\frac{1}{n}$, the APS value is no less than MMS value with any non-negative valuations. 

    Our Theorem \ref{APStheorem6} holds for arbitrary entitlement, so when the agents' entitlements are equal, the existence of a $\frac{1}{2}$-APS allocation also holds and we still have a polynomial-time algorithm to compute it. This implies the existence of $\frac{1}{2}$-MMS allocation and a polynomial-time algorithm. 
\end{proof}

\section{WMMS with General XOS Valuations}

In this section, we first show that APS is no longer an upper bound of WMMS, if agents have asymmetric entitlements. 
In fact, there are instances where WMMS is arbitrarily larger than APS, and thus the algorithms in the previous section do not have any approximation guarantee on WMMS.

\begin{proposition} \label{WMMSAPS}
    For any $\delta > 0$, there exists a fair allocation instance $(N, M, (v_i)_{i\in [n]}, (b_i)_{i\in [n]})$ with additive valuations where there is an agent $a_i\in N$ that
    \begin{equation*}
        \frac{\APS_i}{\WMMS_i} < \delta.
    \end{equation*}
\end{proposition}

\begin{proof}
    We prove by constructing a fair allocation instance $(N, M, (v_i)_{i\in [n]}, (b_i)_{i\in [n]})$ as follows. There are $n$ agents and $m = n$ items to be allocated. Each agent $a_i\in N$ has the same additive valuation
    \begin{equation*}
        v_i(\{g_j\}) = \begin{cases}
            \varepsilon, &j\in [n - 1], \\
            1 - (n - 1)\varepsilon, &j = n,
        \end{cases}
    \end{equation*}
    where $\varepsilon > 0$ is a small value to be determined. The entitlements of the agents are
    \begin{equation*}
        b_1 = b_2 = \cdots = b_{n - 1} = \varepsilon, b_n = 1 - (n - 1)\varepsilon.
    \end{equation*}

    It is obvious that the WMMS partition of all the agents is to give each of the first $n - 1$ agents one item in $\{g_1, g_2, \dots, g_{n - 1}\}$ and to give $a_n$ the item $g_n$. Therefore, the WMMS value of $a_n$ is
    \begin{equation*}
        \WMMS_n = 1 - (n - 1)\varepsilon.
    \end{equation*}
    However, as to the APS value, consider the price
    \begin{equation*}
        p_1 = p_2 = \cdots = p_{n - 1} = 0, p_n = 1.
    \end{equation*}
    Now $a_n$ could not afford the item $g_n$, so
    \begin{equation*}
        \APS_n \le (n - 1)\varepsilon.
    \end{equation*}
    For any $\delta > 0$, take
    \begin{equation*}
        \varepsilon < \frac{1}{n - 1} \frac{\delta}{\delta + 1},
    \end{equation*}
    then for agent $a_n$,
    \begin{equation*}
        \frac{\APS_n}{\WMMS_n} \le \frac{(n-1)\varepsilon}{1-(n-1)\varepsilon} <\delta,
    \end{equation*}
    which proves the proposition.
\end{proof}

\citet{farhadi2019fair} showed the existence of $\frac{1}{n}$-WMMS allocation for agents with general additive valuations and an instance where no allocation better than $\frac{1}{n}$-WMMS approximation exists. In the following, we show that a $\frac{1}{n}$-WMMS allocation exists for agents with general XOS valuations, which extends the previous result. As to the upper bound, we prove that no algorithm can be better than $\frac{1}{n}$-WMMS for binary XOS valuations. However, if all agents have binary additive valuations, an exact WMMS allocation can be found in polynomial time.


\subsection{Computing $\frac{1}{n}$-WMMS Allocations}

For a fair allocation instance $(N,M,(v_i)_{i\in [n]},(b_i)_{i\in [n]})$, without loss of generality, let $b_1\ge b_2\ge \cdots \ge b_n$. 
Denote by $S^i$ the WMMS partition of $a_i$ as in Definition \ref{defWMMSpart}. It is easy to find that 
\begin{equation*}
    v_i(S^i_j) \ge \WMMS_i, \forall j\in [i], i\in [n].
\end{equation*}
Therefore, among the WMMS partition of $a_i$, there are at least $i$ bundles which satisfy $a_i$. This observation leads to Algorithm \ref{WMMSexist}. Intuitively, in Algorithm \ref{WMMSexist}, we adopt a Round-Robin-like process, where every agent obtains an item in a round. Then in the first round, before $a_i$'s turn, only $(i - 1)$ items have been taken. So each agent $a_i$ can choose a complete bundle in $\{S^i_j, j\in [i]\}$ as her target in the first round. Overall, each agent has a bundle that satisfies her WMMS and she is the first one to select items in the bundle during the Round-Robin-like process. With the help of the additive function that maximizes the XOS valuation for the bundle, each agent can get a bundle of at least $\frac{1}{n}$ of her WMMS.

Inheriting the notation of Definition \ref{defWMMSpart}, now we give a concrete description of Algorithm \ref{WMMSexist}.

\begin{algorithm}
\caption{$\frac{1}{n}$-WMMS Existence for General XOS Valuations} \label{WMMSexist}
\begin{algorithmic}[1]
\Require An instance with general XOS valuations $(N,M,(v_i)_{i\in [n]},(b_i)_{i\in [n]})$, the WMMS value $(\WMMS_i)_{i\in [n]}$ and the WMMS partition $(S^i)_{i\in [n]}$.
\Ensure A $\frac{1}{n}$-WMMS allocation $A=(A_1, A_2, \ldots,A_n)$.
\State Rename the agents so that $b_1\ge b_2\ge \cdots\ge b_n$.
\State Initialize $n$ empty bundles $(A_i\gets \varnothing)_{i\in [n]}$.
\State Set $r\gets 1$. \Comment{Round counter.}
\While {$N \neq \varnothing$}
\For {$i\gets 1, 2, \dots, n$}
\If {$r = 1$}
\State Select a bundle $B^i$ from $(S^i_j)_{j\in [i]}$, where $g_1, g_2, \dots, g_{i-1}\notin B^i$. \Comment{The bundle is still complete.} \label{WMMSbundle}
\State Suppose $v_i(S)=\max_{t\in L}l_t(S)$ for $S\subseteq M$, where $(l_t)_{t\in L}$ are additive functions. Let $t_0 \gets \mathop{\arg\max}_{t\in L} l_t(B^i)$ and $l^i \gets l_{t_0}$. \label{WMMSXOS}
\State Set $g_i\gets \mathop{\arg\max}_{e\in B^i} l^i(e)$. 
\EndIf
\If {$a_i\in N$}
\If {$B^i \neq \varnothing$} \Comment{Round-Robin-like process} 
\State Set $g\gets \mathop{\arg\max}_{e\in B^i} l^i(e)$. \label{WMMSroundrobin}
\State Set $A_i\gets A_i\cup \{g\}$.
\State Set $B^j\gets B^j\setminus \{g\}$ for $j\in [n]$, if $B^j$ is already selected.
\State Set $M\gets M\setminus \{g\}$.
\Else
\State Set $N\gets N\setminus \{a_i\}$. \Comment{Agent $a_i$ finishes.}
\EndIf
\EndIf
\EndFor
\State Set $r\gets r + 1$.
\EndWhile
\If {$M\neq \varnothing$}
\State Allocate the remaining goods arbitrarily.
\EndIf
\Return $A = (A_1, A_2, \dots, A_n)$.
\end{algorithmic}
\end{algorithm}

\begin{theorem}\label{WMMStheorem1}
    For any fair allocation instance $(N, M, (v_i)_{i\in [n]}, (b_i)_{i\in [n]})$ with XOS valuations, there exists a $\frac{1}{n}$-WMMS allocation. 
\end{theorem}

\begin{proof}
    Algorithm \ref{WMMSexist} admits the existence of $\frac{1}{n}$-WMMS allocation. First we claim that
    \begin{equation*}
        v_i(S^i_j)\ge \WMMS_i, \forall j\in [i], i\in [n].
    \end{equation*}
    Recall the definition \ref{defWMMSpart} of WMMS partition, the WMMS partition of $a_i$ is
    \begin{equation*}
        S^i = \mathop{\arg\max}_{S\in \Pi_n} \min_{j\in [n]} v_i(S^i_j) \frac{b_i}{b_j},
    \end{equation*}
    and $S^i_j$ is the bundle $a_i$ wants to assign to $a_j$. By the definition \ref{defWMMS} of WMMS,
    \begin{equation*}
        v_i(S_j^i) \frac{b_i}{b_j} \ge \WMMS_i, \forall j\in [n].
    \end{equation*}
    With the assumption that $b_1\ge b_2\ge \cdots\ge b_n$, we have $v_i(S_j^i)\ge \WMMS_i$ for $j\in [i]$.
    
    This means that every agent $a_i$ starts from a bundle $B^i$ where $v_i(B^i)\ge \WMMS_i$ in line \ref{WMMSbundle}. As each agent could only get one item in a round, before the first round of $a_i$, only $(i - 1)$ items have been taken away by the previous agents. Therefore, the $i$ bundles $(S^i_j)_{j\in [i]}$ for $a_i$ could guarantee the existence of $B^i$.
    
    As to the approximation ratio, the agent $a_i$ is the first one to take goods from the bundle $B^i$, and she will always select the most valuable one according to the additive function $l^i$ as in line \ref{WMMSXOS}. Let $l^i|_{B^i}$ be an additive function where
    \begin{equation*}
        l^i|_{B^i}(e) = \begin{cases}
            l^i(e), &e\in B^i, \\
            0, &e\notin B^i.
        \end{cases}
    \end{equation*}
    Suppose $a_i$ receives some item $g^r$ at the $r$-th round, then by the selection of line \ref{WMMSroundrobin}, we have $l^i|_{B^i}(g^r) \ge l^i|_{B^i}(e^r)$ for all items $e^r$ that any other agent could get between the $r$-th round and the $(r + 1)$-th round of $a_i$. As no agent takes any item from $B^i$ until $a_i$ does in the first round and there are $n$ agents in total, this process indicates
    \begin{equation*}
        l^i (A_i) = l^i|_{B^i} (A_i) \ge \frac{1}{n} l^i|_{B^i} (B^i) = \frac{1}{n} l^i (B^i).
    \end{equation*}
    By the definition of XOS functions, $v_i(S)\ge l^i(S), \forall S\subseteq M$. In summary, we have
    \begin{equation*}
        v_i(A_i)\ge l^i (A_i)\ge \frac{1}{n} l^i (B^i) = \frac{1}{n} v_i(B^i) \ge \frac{1}{n} \WMMS_i.
    \end{equation*} 
    Therefore, the allocation $A$ is a $\frac{1}{n}$-WMMS allocation for the instance.
\end{proof}

\subsection{Upper Bound}

The tight example exploits the fact that the agent with the largest entitlement may only have one bundle satisfying her WMMS.

\begin{theorem}\label{WMMSexample}
    For any $n\ge 2$, there exists a fair allocation instance of $n$ agents with binary XOS valuations where no allocation is better than $\frac{1}{n}$-WMMS.
\end{theorem}

\begin{proof}
    Let $n$ be the number of agents, we construct a fair allocation instance that no more than $\frac{1}{n}$-WMMS allocation exists. Consider an instance $(N, M, (v_i)_{i\in [n]}, (b_i)_{i\in [n]})$ as follows. Here we set $m = 2n - 1$ and $M=\{g_1,g_2,\ldots,g_{2n-1}\}$. For any $S\subseteq M$, the value functions of agents are
    \begin{align*}
        v_1(S) =& v_2(S) = \cdots = v_{n - 1}(S) \notag \\
        =& \max\{|S\cap\{g_1\}|, |S\cap\{g_2\}|, \dots, |S\cap\{g_n\}|\}, \\
        v_n(S) =& \max\{|S\cap\{g_1, g_2, \dots, g_n\}|, \notag \\
        &|S\cap\{g_{n + 1}\}|, |S\cap\{g_{n + 2}\}|, \dots, |S\cap\{g_{2n - 1}\}|\}. 
    \end{align*}
    The entitlements of the agents are 
    \begin{equation*}
        b_1 = b_2 = \cdots = b_{n - 1} = \frac{1}{2n - 1}, b_n = \frac{n}{2n - 1}.
    \end{equation*}

    It is straightforward to observe that the valuations are binary XOS. We then calculate the WMMS value of the $n$ agents. In spite of the imbalance in the entitlement, the only reasonable allocation for $(a_i)_{i\in [n - 1]}$ is to give each of the $n$ agents one item in $\{g_1, g_2, \dots, g_n\}$ and allocate the remaining items arbitrarily. Otherwise, there will be at least one agent receiving a zero-valued bundle, which makes her minimum value in WMMS definition be $0$. By definition \ref{defWMMS},
    \begin{equation*}
        \WMMS_1 = \WMMS_2 = \cdots = \WMMS_{n - 1} = \frac{1}{n}.
    \end{equation*}
    
    For $a_n$, since she has $n$ times as much entitlement as the others, according to her value function, her favorable allocation is to receive $\{g_1, g_2, \dots, g_n\}$ herself and give every other agent an item in $\{g_{n + 1}, g_{n + 2}, \dots, g_{2n - 1}\}$. Therefore, 
    \begin{equation*}
        \WMMS_n = n.
    \end{equation*}

    As to the allocation, to ensure that the value every agent receives is greater than $0$, each of the first $n - 1$ agents must obtain at least an item in $\{g_1, g_2, \dots, g_n\}$. Without loss of generality, we will assign the items $\{g_1, g_2, \ldots, g_{n-1}\}$ to the first $n-1$ agents, which makes agent $n$ could only receive a bundle $S\subseteq\{g_{n}, g_{n+1}, \dots, g_{2n-1}\}$ with value $1$. In summary, either for some agent in $(a_i)_{i\in [n - 1]}$, the allocation is $0$-WMMS, or for $a_n$ the allocation is at most $\frac{1}{n}$-WMMS, which finishes the proof.
\end{proof}

\subsection{Binary Additive Valuations}

In contrast to the $\frac{1}{n}$-WMMS bound for general additive valuations \cite{farhadi2019fair} and binary XOS valuations, we present a polynomial-time algorithm which outputs an exact WMMS allocation for agents with binary additive valuations. For a fair allocation instance $(N,M,(v_i)_{i\in [n]},(b_i)_{i\in [n]})$ with binary additive valuations, let
\begin{equation*}
    D^i = \{g\in M| v_i(g) = 1\},
\end{equation*}
then the valuation of agent $a_i$ is $v_i(S) = |S\cap D^i|, \forall S\subseteq M$. 

We begin with Algorithm \ref{WMMSpartAddAlgo}, which computes the WMMS partition, hence the WMMS value, of any agent in polynomial time.

\begin{algorithm}
\caption{WMMS Partition for Binary Additive Valuations} \label{WMMSpartAddAlgo}
\begin{algorithmic}[1]
\Require An instance with binary additive valuations $(N,M,(v_i)_{i\in [n]},(b_i)_{i\in [n]})$ and an agent $a_i\in N$.
\Ensure A WMMS partition $S^i = (S^i_j)_{j\in [n]}$ of agent $a_i$.
\State Initialize $n$ empty bundles $(S^i_j \gets \varnothing)_{j\in [n]}$.
\State Let $D^i \gets \{g\in M| v_i(g) = 1\}$.
\For {$g\in D^i$}
\State Set $j_{\min} \gets \mathop{\arg\min}_{j\in [n]} \frac{|S^i_j|}{b_j}$, break tie arbitrarily.
\State Set $S^i_{j_{\min}} \gets S^i_{j_{\min}}\cup\{g\}$.
\EndFor
\State Allocate goods in $M\setminus D_i$ arbitrarily.
\State \Return $S^i = (S^i_1, S^i_2, \dots, S^i_n)$.
\end{algorithmic}
\end{algorithm}

\begin{lemma}\label{WMMSpartAddThm}
    For any fair allocation instance $(N,M,(v_i)_{i\in [n]},(b_i)_{i\in [n]})$ with binary additive valuations, Algorithm \ref{WMMSpartAddAlgo} returns a WMMS partition of agent $a_i$ in polynomial time. 
\end{lemma}

\begin{proof}
    After the $r$-th step of the for loop, denote the $n$ bundles we assigned are $S^{i, r} = (S^{i, r}_j)_{j\in [n]}$, then the set of items we have assigned are $M^{i, r} = \cup_{j\in[n]} S^{i, r}_j$. 
    
    Let $L^{i,r}$ denote the bundle whose value to agent $a_i$, divided by the entitlement $b_j$ corresponding to its index $j$, is the smallest. That is

    \begin{equation*}
        L^{i, r} = \min_{j\in [n]} \frac{|S^{i, r}_j|}{b_j}.
    \end{equation*}
    We claim that for all the for loop  index $r$, $L^{i, r}$ satisfies
    \begin{equation} \label{WMMSpartAddIneq}
        \frac{|S^{i, r}_j| - 1}{b_j}\le L^{i,r}\le \frac{|S^{i, r}_j|}{b_j}, \forall j\in [n].
    \end{equation}
    
    The right side of inequality \ref{WMMSpartAddIneq} definitely holds according to the definition of $L^{i,r}$. We will prove the left side of inequality \ref{WMMSpartAddIneq} by contradiction. We assume that there is a bundle $S_k^{i,r}$ satisfies $\frac{|S^{i, r}_k| - 1}{b_k}> L^{i,r}$ and that the last item in bundle $S_{k}^{i,r}$ was assigned to it during the $r^\prime$-th for loop. After $r^\prime$-th round, there is
    \begin{align*}
        \frac{|S^{i, r^\prime}_k|}{b_k}=\frac{|S^{i, r}_k| - 1}{b_k}> L^{i,r}\ge L^{i,r^\prime}.
    \end{align*}
    The last inequality holds since, according to Algorithm \ref{WMMSpartAddAlgo}, the number of items in each bundle is non-decreasing, and then the value of $L^{i,r}$ is non-decreasing in $r$. Since $\frac{|S^{i, r^\prime}_k|}{b_k}$ is strictly larger than $L^{i,r^\prime}$, in this round, the item will not assigned to $S^{i,r^\prime}_{k}$ according to Algorithm \ref{WMMSpartAddAlgo}, which contradicts our assumption. We have proved the left side of Inequality \ref{WMMSpartAddIneq}.
    
    Next, we will show that for any for loop index $r\in [|D^i|]$, the set of bundles $S^{i, r}$ is a WMMS partition for the set of assigned items $M^{i, r}$ for agent $a_i$. Therefore, by the definition of WMMS, we have
    \begin{equation*}
        \WMMS_i(M^{i, r}) = L^{i, r} b_i.
    \end{equation*}
    
    We will also prove this by contradiction. We assume that partition $S^{i,r}$ is not a WMMS partition for item set $M^{i,r}$. Let a WMMS partition be $T^{i, r} = (T^{i, r}_j)_{j\in [n]}$. 
    Therefore, by the definition of WMMS partition, there is
    \begin{equation*}
        \min_{j\in [n]} \frac{|T^{i, r}_j|}{b_j} > L^{i, r}.
    \end{equation*}
    
    Since $T^{i,r}$ is not equal to $S^{i,r}$, there must exist a bundle index $k$ where
    \begin{align*}
        |T^{i, r}_{k}| \le |S^{i, r}_{k}| - 1.
    \end{align*}
    According to Inequality \ref{WMMSpartAddIneq}, it means
    \begin{equation*}
        \frac{|T^{i, r}_{k}|}{b_k} \le \frac{|S^{i, r}_{k}| - 1}{b_k}\le L^{i, r},
    \end{equation*}
    which leads a contradiction.

    This finishes the proof for the correctness of Algorithm \ref{WMMSpartAddAlgo}. 
    
    For the time complexity of Algorithm \ref{WMMSpartAddAlgo}, we need to iterate through all the $m$ items, and each iteration costs a binary search in $n$ values. Overall, the time complexity is $O(m\log n)$.
\end{proof}

It is straightforward to observe that the calculation of the WMMS partition serves as the ``worst'' case allocation for each agent. That is, if we let each agent select items in her own favor according to the rule in Algorithm \ref{WMMSpartAddAlgo}, then from the point of view of a certain agent, there will be a possibility that she is indifferent to some items taken by the other agents. This observation leads to Algorithm \ref{WMMSAddAlgo}.

\begin{algorithm}
\caption{WMMS Allocation for Binary Additive Valuations} \label{WMMSAddAlgo}
\begin{algorithmic}[1]
\Require An instance with binary additive valuations $(N,M,(v_i)_{i\in [n]},(b_i)_{i\in [n]})$.
\Ensure A WMMS partition $A=(A_1, A_2, \ldots,A_n)$.
\State Initialize $n$ empty bundles $(A_i \gets \varnothing)_{i\in [n]}$.
\State Let $D^i \gets \{g\in M| v_i(g) = 1\}, \forall i\in [n]$.
\While {$N\neq \varnothing$}
\State Set $i_{\min} \gets \mathop{\arg\min}_{i\in [n]} \frac{|A_i|}{b_i}$.
\If {$D^{i_{\min}}\neq \varnothing$}
\State Select an arbitrary item $g\in D^{i_{\min}}$ and set $A_i\gets A_i\cup\{g\}$.
\State Set $D^i\gets D^i\setminus\{g\}, \forall i\in [n]$.
\State Set $M\gets M\setminus\{g\}$.
\Else
\State $N\gets N\setminus \{a_{i_{\min}}\}$.
\EndIf
\EndWhile
\State Allocate remaining goods in $M$ arbitrarily.
\State \Return $S^i = (S^i_1, S^i_2, \dots, S^i_n)$.
\end{algorithmic}
\end{algorithm}

\begin{theorem}
Algorithm \ref{WMMSAddAlgo} returns a WMMS allocation in polynomial time for any fair allocation instance $(N, M, (v_i)_{i\in [n]}, (b_i)_{i\in [n]})$ with binary additive valuations.
\end{theorem}

\begin{proof}
For any agent $a_i$, we have $v_i(A_i)\ge \WMMS_i$. That is because if everybody takes the item $g$ from the set $D_i$ in each while loop, Algorithm \ref{WMMSAddAlgo} does the same work as Algorithm \ref{WMMSpartAddAlgo}. In this case, Lemma \ref{WMMSpartAddThm} have proved that $v_i(A_i)\ge \WMMS_i$. Otherwise, if some agent gets an item not in $D_i$, the remaining items in $D_i$ will be increased, and the number of items will be increased at the same time. We still have $|A_i|=v_i(A_i)\ge \WMMS_i$. The time complexity of Algorithm \ref{WMMSAddAlgo} is also $O(m\log n)$.
\end{proof}

\section{Conclusion Remarks}
In this paper, we prove the tight approximation guarantees for asymmetric fair allocation with binary XOS valuations under APS and WMMS fairness criteria.
Specifically, we design a polynomial-time algorithm that computes a $\frac{1}{2}$-approximate APS allocation, which matches the upper bound of the approximation ratio.  This result also implies the existence of a 
$\frac{1}{2}$-approximate MMS allocation when agents have equal entitlements. 
When agents have different entitlements, we show that APS can be arbitrarily smaller than WMMS, even when agents have additive valuations.
We then design an algorithm that ensures a $\frac{1}{n}$-WMMS allocation and prove that a better than $\frac{1}{n}$ approximation ratio is not guaranteed. The approximation ratio holds even for general XOS valuations.
In addition, we show that an exact WMMS allocation can be achieved for binary additive valuations; this result serves to highlight that obtaining a $\frac{1}{n}$-WMMS guarantee under binary XOS valuations is non-trivial.
There are several interesting future research directions. 
First, it is still unknown if constant approximations can be ensured if agents have subadditive valuations with binary marginals.
Second, for $\APS$/$\WMMS$ with asymmetric agents under binary submodular (matroid-rank) valuations, while our binary XOS results already apply, it remains open whether strictly better approximation guarantees can be achieved beyond the symmetric-agent setting.
Third, in the current work, we have solely focused on the allocation of goods. It is interesting to investigate the minor problem of chores.






\section*{Acknowledgments}

We thank Yiwei Gao and Ankang Sun for helpful discussions.



\bibliographystyle{ACM-Reference-Format} 
\bibliography{ref}


\end{document}